\documentclass[conference]{IEEEtran}

\usepackage[utf8]{inputenc}
\usepackage{verbatim}
\usepackage{optidef}

\usepackage{amsmath,amssymb,amsthm,amsfonts,mathrsfs}
\usepackage{tablefootnote}
\usepackage[margin=0.5in]{geometry}

\usepackage{hyperref} 

\usepackage{graphicx,color,cases} 
\usepackage{xfrac}  
\usepackage{xargs}

\usepackage{tikz}	
\usetikzlibrary{backgrounds,fit,decorations.pathreplacing}  
\usetikzlibrary{automata} 

\usepackage[colorinlistoftodos,prependcaption]{todonotes}

\definecolor{DarkGreen}{rgb}{0.1,0.5,0.1}
\definecolor{DarkRed}{rgb}{0.5,0.1,0.1}
\definecolor{DarkBlue}{rgb}{0.1,0.1,0.5} 
  
\def\>{\rangle} 
\def\<{\langle}

\newtheorem{definitionenv}{Definition}
\newtheorem{lemmaenv}[definitionenv]{Lemma}
\newtheorem{theoremenv}[definitionenv]{Theorem}
\newtheorem{corollaryenv}[definitionenv]{Corollary}
\newtheorem{propositionenv}[definitionenv]{Proposition}
\newtheorem{conjectureenv}[definitionenv]{Conjecture}
\newtheorem{exampleenv}{Example}
\newtheorem{app-lemmaenv}[section]{Lemma}

\newenvironment{definition}{\begin{definitionenv}\rm}{\end{definitionenv}}
\newenvironment{lemma}{\begin{lemmaenv}\rm}{\end{lemmaenv}}
\newenvironment{theorem}{\begin{theoremenv}\rm}{\end{theoremenv}}
\newenvironment{corollary}{\begin{corollaryenv}\rm}{\end{corollaryenv}}
\newenvironment{example}{\begin{exampleenv}\rm}{\end{exampleenv}}
\newenvironment{proposition}{\begin{propositionenv}\rm}{\end{propositionenv}}
\newenvironment{conjecture}{\begin{conjectureenv}\rm}{\end{conjectureenv}}
\newenvironment{app-lemma}{\begin{app-lemmaenv}\rm}{\end{app-lemmaenv}}

\newcommand{\bd}{\begin{definition}}
\newcommand{\ed}{\end{definition}}
\newcommand{\bl}{\begin{lemma}}
\newcommand{\el}{\end{lemma}}
\newcommand{\elp}{\hspace*{\fill} $\Box$
                 \end{lemma}}
\newcommand{\bt}{\begin{theorem}}
\newcommand{\et}{\end{theorem}}
\newcommand{\etp}{\hspace*{\fill} $\Box$
                 \end{theorem}}
\newcommand{\bc}{\begin{corollary}}
\newcommand{\ec}{\end{corollary}}
\newcommand{\ecp}{\hspace*{\fill} $\Box$
                 \end{corollary}}
\newcommand{\bcj}{\begin{conjecture}}
\newcommand{\ecj}{\end{conjecture}}

\newcommand{\be}{\begin{example}}
\newcommand{\ee}{\end{example}}
\newcommand{\eep}{\hspace*{\fill} $\Box$
                 \end{example}}
\newcommand{\bp}{\begin{proposition}}
\newcommand{\ep}{\end{proposition}}
\newcommand{\epp}{
                 \end{proposition}}

\newcommand{\bx}{{\bf x}}

\newcommand{\cG}{{\cal G}}

\newcommand{\cK}{{\cal K}}

\newcommand{\cN}{{\cal N}}

\newcommand{\cQ}{{\cal Q}}

\newcommand{\mC}{{\mathbb C}}

\newcommand{\ket}[1]{|#1\rangle}

\newcommand{\tr}{\text{tr}}
\newcommand{\Tr}[1]{\text{Tr}\left(#1\right)}
\newcommand{\wt}[1]{\text{wt}\left(#1\right)}

\newcommand{\ot}{\otimes}

\newcommand{\psiew}{\ket{{\texttt{AUX}}}}

\newcommand{\supewe}{}

\newcommand{\supwe} {^\texttt{we}}

\newcommand{\ssl}{^{\texttt{SL}}}

\newcommand{\zo}{\{0,1\}}

\newcommandx{\yellownote}[2][1=]{\todo[linecolor=yellow,backgroundcolor=yellow!25,bordercolor=yellow,#1]{#2}}
\newcommandx{\greennote}[2][1=]{\todo[inline,linecolor=olive,backgroundcolor=green!25,bordercolor=olive,#1]{#2}}

\begin{document}
\title{Linear programming bounds for quantum amplitude damping codes}
\author{%
  \IEEEauthorblockN{Yingkai Ouyang}
  \IEEEauthorblockA{Department of Physics \& Astronomy\\
  University of Sheffield\\
  Sheffield, S3 7RH, United Kingdom\\
                    Email:  y.ouyang@sheffield.ac.uk}
  \and
  \IEEEauthorblockN{Ching-Yi Lai}
  \IEEEauthorblockA{Institute of Communications Engineering\\
			National Chiao Tung University\\
			Hsinchu 30010, Taiwan\\
                    Email: cylai@nctu.edu.tw}
}

\maketitle

\begin{abstract}
Given that approximate quantum error-correcting (AQEC) codes have a potentially better performance than perfect quantum error correction codes, it is pertinent to quantify their performance.
While quantum weight enumerators establish some of the best upper bounds on the minimum distance of quantum error-correcting codes, these bounds do not directly apply to AQEC codes.
Herein, we introduce quantum weight enumerators for amplitude damping (AD) errors and work within the framework of approximate quantum error correction. 
In particular, we introduce an auxiliary exact weight enumerator that is intrinsic to a code space and moreover, we establish a linear relationship between the quantum weight enumerators for AD errors and this auxiliary exact weight enumerator. 
This allows us to establish a linear program that is infeasible only when AQEC AD codes with corresponding parameters do not exist. 
To illustrate our linear program, we numerically rule out the existence of three-qubit AD codes that are capable of correcting an arbitrary AD error.
\end{abstract}


\section{Introduction}
The distance of an error-correcting code is of central importance in coding theory, because it quantifies the number of adversarial errors that can be corrected. 
For codes of fixed length and rate, upper and lower bounds on their distance can be determined.
The best lower bounds can be obtained from various randomized code constructions that yield the Gilbert-Varshamov bound~\cite{MS77} and this is also true in the quantum case~\cite{FeM04,JiX11,ouyang2014concatenated}).
On the contrary, markedly different techniques are used to derive upper bounds. 
In classical coding theory, weight enumerators count the weight distribution of codewords in a code~\cite{MS77}. 
The MacWilliams identity establishes a linear relationship between the weight enumerators of a code and that of its dual code. 
This allows one to obtain upper bounds on the distance of codes. Further extensions of this technique leads to the celebrated linear programming bounds \cite{Aal90} and their generalizations \cite{schrijver2005new}.

The notion of weight enumerators in the quantum setting is less obvious, because quantum codes on $n$ qubits are subspaces of $\mathbb C^{2^n}$, and these subspaces do not in general admit a combinatorial interpretation.
Shor and Laflamme nonetheless introduced a meaningful definition for weight enumerators for quantum codes \cite{ShL97} in terms of the codes' projectors $P$ and a nice error basis for matrices.
In particular, the Shor-Laflamme (SL) quantum weight enumerators are sums of terms of the form $|\tr(EP)|^2$ and $\tr(E P E^\dagger P)$, respectively, where the sums are performed over all Paulis $E$ of a given weight.
We will call the vector of these enumerators labeled by Pauli weights the A-type and B-type quantum weight enumerators, respectively.
Shor and Laflamme showed that the A- and B-type quantum weight enumerators are still linearly related in a way reminiscent of the classical relationship~\cite{SL97}. The relation between the two enumerators is the quantum analogue of the famous MacWilliams identity. 
Variations on the SL enumerators were then studied by Rains, which allowed better bounds on the parameters of quantum codes~\cite{Rains98a}.
Because of the existence of a linear relationship between the two types of enumerators, 
linear programming techniques can be applied to establish upper bounds on the minimum distance for (small) quantum stabilizer codes~\cite{CRSS98}. Algebraic linear programming bounds based on the MacWilliams identity, such as the Singleton, Hamming, and the first linear programming bounds, are also derived for general quantum codes~\cite{AL99}. These results have been extended to entanglement-assisted quantum stabilizer codes~\cite{LBW13,LA18} and quantum data-syndrome codes~\cite{ALB20}.
Also there is a MacWilliams identity for (entanglement-assisted) quantum convolutional codes~\cite{LHL16}.

Although the distance of a quantum code is a meaningful metric with respect to adversarial noise, estimates on the performance of a quantum code derived from the distance under specific noise models are often overly pessimistic.
For instance, while a minimum of five qubits is needed to perfectly correct an arbitrary error~\cite{ShL97}, four qubits suffice to correct a single \textit{amplitude damping}  (AD) error~\cite{LNCY97}.
Quantum codes designed specifically to combat AD errors are called \textit{AD codes} and are well-studied~\cite{LNCY97,CLY97,SSSZ09, DGJZ10,jackson2016concatenated,ouyang2014permutation,ouyang2015permutation,grassl2018quantum,ouyang2019permutation}.
However, existing quantum weight enumerators give no direct result regarding limits on the ultimate performance of AD codes. 
To better understand these fundamental limits, it would be advantageous to have MacWilliams-type identities for different quantum weight enumerators defined for various noisy quantum channels.
In spite of this, finding suitable generalizations of the quantum weight enumerators to specialized noise models remains an open problem.

There are several challenges in generalizing linear programming bounds for quantum codes to allow the consideration of AD errors.
First, quantum weight enumerators only describe quantum error correction in the perfect setting~\cite{Rains98a}, and therefore cannot describe quantum codes designed for the AD channel that often approximately satisfy the Knill-Laflamme perfect quantum error correction conditions~\cite{KnL97}.
Second, while the Pauli errors used to express quantum weight enumerators form a nice error basis, the Kraus operators of an AD channel do not span the matrix space on all of the qubits.

In this work, we overcome the above challenges, and extend the theory of quantum weight enumerators to deal with approximate quantum error-correcting (AQEC) codes and AD errors.
Namely, we generalize the two SL quantum weight enumerators to address quantum codes that approximately correct AD errors. 
While we do not have a MacWilliams identity that establishes a direct linear relationship between these two generalized quantum weight enumerators, we do establish an indirect linear relationship between them.
To enable this, we introduce an \textit{auxiliary exact weight enumerator} with respect to Pauli operators,
which is exact in the sense that it depends explicitly on the matrix decomposition of the code projector $P$ in the Pauli basis.
We thereby show linear connections between this enumerator and our two generalized quantum weight enumerators. 
This allows us to establish a linear program that is infeasible only when AQEC AD codes do not exist. 
To illustrate our linear program, we numerically rule out the existence of  three-qubit AQEC AD codes   that are capable  of  correcting  an  arbitrary  AD  error.
Our linear program cannot eliminate the existence of a four-qubit code that can correct one AD error and this agrees to the four-qubit AD code proposed in~\cite{LNCY97}. 

This paper is organized as follows.
In the  next section, we review notation for Pauli operators, amplitude damping channels, and review quantum weight enumerators. 
In Sec.~\ref{sec:AD-wtenum}, we introduce our quantum weight enumerators specialized to deal with AD errors and approximate quantum error correction.
In Sec.~\ref{sec:aux}, we introduce auxiliary weight enumerators, and in Sec.~\ref{sec:connection-matrices}, we propose connection matrices that establish linear relationships between our quantum weight enumerators and the auxiliary weight enumerators.
In Sec.~\ref{sec:linear program bounds}, we formulate a linear program that is infeasible only when the corresponding AD-code does not exist.
We conclude our results and discuss the potential for further work in Sec.~\ref{sec:discussions},

\section{Preliminaries} \label{sec:prelims}
 
\subsection{Pauli Operators}
A single-qubit state space is a two-dimensional complex Hilbert space $\mathbb{C}^2$,
and a multiple-qubit state space is simply the tensor product space of single-qubit spaces.
The Pauli matrices
$$I_2=\begin{bmatrix}1 &0\\0&1\end{bmatrix}, X=\begin{bmatrix}0 &1\\1&0\end{bmatrix},  Z=\begin{bmatrix}1 &0\\0&-1\end{bmatrix}, Y=iXZ$$
form a basis of the linear operators on a single-qubit state space.
Let  
\[{\cG}_n=\{M_1\otimes M_2\otimes \cdots \otimes M_n: M_j\in\{I_2,X,Y,Z\} \},\] which is a basis of the linear operators on the $n$-qubit state space $\mathbb{C}^{2^n}$.
We denote the weight of any element of $\cG_n$
as $\wt{E}$, which is the number of $M_j$'s that are non-identity matrices.

\subsection{Amplitude Damping Channel}

Amplitude damping (AD) errors model energy relaxation in quantum harmonic oscillator systems and photon loss in photonic systems.
By ensuring that each quantum harmonic oscillator couples identically to a unique bosonic bath, in the low temperature limit,
the effective noise model can be described by an AD channel.
When quantum information lies in a qubit, 
the corresponding AD channel 
$\cN_\gamma$ models energy loss in a two-level system, where $\gamma$ is the probability that an excited state relaxes to the ground state, and $\cN_\gamma$ has two Kraus operators $A_0$ and $A_1$, where
\[
A_0=\begin{bmatrix}1&0\\0&\sqrt{1-\gamma}\end{bmatrix}, \quad A_1=\begin{bmatrix}0&\sqrt{\gamma}\\ 0&0\end{bmatrix}.
\]
When energy loss occurs independently and identically in an $n$-qubit system, 
the corresponding noisy channel can be modeled 
with $\cN_{n,\gamma}=\cN_{\gamma}^{\otimes n}$.
The set of all Kraus operators of $\cN_{n,\gamma}$ can be written as 
\begin{align}
\mathcal K =  \{A_{\bx}\triangleq A_{x_1}\otimes \cdots \ot A_{x_n}: \bx\in \zo^n\}.   
\end{align}
Since the Kraus operator $A_1$ models energy loss on one qubit, 
it is useful to know how many times the Kraus operator $A_1$ occurs in $A_\bx$.
Hence we define the following property of $A_\bx$.

\bd The  weight of $A_{\bx}$ is defined as $\wt{\bx}$.
\ed
The weight of $A_\bx$ counts the number of qubits where $A_\bx$ induces energy loss.
For example, $\wt{A_1\otimes A_0\ot A_1}=\wt{A_{101}}= 2$, which corresponds to energy loss in two qubits.
Using this notion of weight, we partition the set of Kraus operators $\cK$ accordingly. Namely, by denoting 
\begin{align}
    \cK_i =  \{ E \in \cK : \wt{E} = i\},
\end{align}
we have $\cK = \cK_0 \cup \dots \cup \cK_n$.
In this terminology, a code corrects $t$ errors perfectly if all the errors in $\cK_i$ for $i \le t$ satisfy the Knill-Laflamme quantum error correction criterion \cite{KnL97}.

\subsection{Quantum Codes and Weight Enumerators}
An $((n,M))$ quantum code $\cQ$ is an $M$-dimensional subspace of $\mC^{2^n}$.
Let $P$ denote the codespace projector of $\cQ$.
Shor and Laflamme define two  weight enumerators $\{A\ssl_{i}\}$ and $\{B\ssl_{i}\}$ of  $\cQ$  by
\begin{align}
A\ssl_{i}=&\frac{1}{M^2}\sum_{ E\in {\cG}_n, \wt{E}=i } \Tr{E P}\Tr{E^\dag P}, \label{eq:Bij}
\end{align}
and
\begin{align}
B\ssl_{i}=&\frac{1}{M}\sum_{ E\in {\cG}_n, \wt{E}=i} \Tr{E PE^\dag P}, \label{eq:Bij_perp}
\end{align}
for $i=0,\dots,n$ \cite{SL97}.
Note that ${\cG}_n$ is a basis for the linear operators on $\mathbb{C}^{2^n}$.
These two weight enumerators will be called \emph{SL enumerators} in this article.
The power of the SL enumerators is that the perfect quantum error correction criterion of Knill and Laflamme are equivalent to certain linear constraints on these SL enumerators.
From this context, perturbations to the Knill-Laflamme quantum error correction criterion can be understood by directly perturbing linear constraints on quantum weight enumerators.

\section{Quantum weight enumerators for amplitude damping errors}
\label{sec:AD-wtenum}

In what follows, we   generalize the SL enumerators to allow direct consideration of AD errors. The above-mentioned key technical difficulty is inability of the corresponding set of Kraus operators $\cal K$ to span the space of linear operators on $\mathbb{C}^{2^n}$.
We nonetheless can generalize SL enumerators to deal with AD channels. Our new enumerators are vectors with coefficients
\begin{align}
    A_i    &= 
    \frac{1}{(\tr P)^2}
    \sum_{E \in \cK_i} \tr(EP)\tr(E^\dagger P),\ i=0,\dots,n, \\ 
    B_i     &= 
    \frac{1}{\tr P}
    \sum_{E \in \cK_i} \tr(EPE^\dagger P) , \ i=0,\dots,n.
\end{align}
It can be shown, as in \cite{AL99} that 
\begin{align}
B_i\geq A_i, \ i=0,\dots,n. \label{eq:BgeqA}
\end{align}
Furthermore, since the code projector $P$ is Hermitian and we have the cyclic property of the trace, it is clear that
$\tr(EP)\tr(E^\dagger P) = \tr(EP)\tr(P^\dagger E^\dagger) = |\tr(EP)|^2$. This implies that $A_i$ is always a sum of non-negative terms, and hence we must have
\begin{align}
    A_i \ge 0 , \quad 0 \le i \le n.
\end{align}
Now the sum of B type enumerators retains
interpretation as the fidelity
of a quantum code after the action of the AD channel without error correction.
Hence
\begin{align}
    B_0 + \cdots + B_n \le 1.
\end{align}
Since the only Kraus operator in $\cK_0$ has a minimum singular value of $(1-\gamma)^{n/2}$ for $A_0$, we have the lower bound
\begin{align}
    A_0 \ge (1-\gamma)^n.
\end{align}
This is reminiscent of the scenario for SL weight enumerators, where we have $A\ssl_0=1.$

Furthermore, it is easy to see that every $B_i$ is at most of order $O(\gamma^i)$, Since the operator norm of Kraus operators from $\cK_i$ is $\gamma^{i/2}$, the operator norm of $EPE^\dagger$ for any $E \in \cK_i$ is at most $\gamma^i \tr(P)$. Hence it follows from the H\"older inequality on the Hilbert-Schmidt inner product that 
\[|\tr(EPE^\dagger P )|= |\<EPE^\dagger ,P \>|\le \|EPE^\dagger \|\|P\|_1,\]
where $\|\cdot \|_1$ denotes the trace norm and $\|\cdot \|$ denotes the operator norm, which is the maximum singular value of a matrix.
Thus by counting the number of terms in $\cK_i$, we have
\begin{align}
    B_i/\gamma^{i} \le \binom n i. 
\end{align}
In what follows, we use the Dirac ket notation to represent weight enumerators as in \cite{LHL16}.
Let $\{|0\>,\dots, |n\>\}$ be an orthonormal basis of $\mathbb{R}^{n+1}$. 
The SL enumerators of $\cQ$ for AD channels are 
\begin{align}
|A \>  &=  A_0 |0\> + \dots + A_{n} |n\>, \notag \\
|B \>  &=  B_0 |0\> + \dots + B_{n} |n\> .
\end{align}
In this paper, we define approximate quantum error correction using the language of quantum weight enumerators. In the case of perfect quantum error correction, for a quantum code that has minimum distance $d$, we must have 
\[
B\ssl_i - A\ssl_i = 0 , \quad i = 0, \dots, d-1.
\]
These equations can be relaxed to yield the following definition of approximate quantum error correction for AD channels.
\bd  \label{def:tc-AD-code}
An $((n,M))$ quantum code is called a $(t,c)$-AD code if
its quantum weight enumerators satisfy the constraints
\begin{align}
B_i-A_i\leq c \gamma^{t+1}, i=0,\dots,t. \label{eq:BleqA}
\end{align}
\ed

\be
The four-qubit code in~\cite{LNCY97} has two logical codewords
\begin{align*}
\ket{0}_L=& \frac{1}{\sqrt{2}}\left(\ket{0000}+\ket{1111}\right),\\
\ket{1}_L=& \frac{1}{\sqrt{2}}\left(\ket{0011}+\ket{1100}\right).
\end{align*}
It has weight enumerators
\begin{align*}
A_0&=\gamma^4/64 - \gamma^3/4 + 5 \gamma^2/4 - 2 \gamma + 1;\\
A_1&=A_2=A_3=0;\\
A_4&=\gamma^4/64.
\end{align*}
\begin{align*}
B_0&=\gamma^4/16 - \gamma^3/4 + 5 \gamma^2/4 - 2 \gamma + 1;\\
B_2&=3 \gamma^4/8 - 3 \gamma^3/4 + 3 \gamma^2/4;\\
B_4&=\gamma^4/16;\\
B_1&=B_3=0.
\end{align*}
Therefore, this code cannot be a $(2,c)$-AD code for any $c>0$,
and is consistent with the fact that this code corrects a single AD error \cite{LNCY97}.

\ee

\be The weight enumerators of the nine-qubit Shor code are as follows.
Note that (\ref{eq:BgeqA}) holds here. In addition, by Definition~\ref{def:tc-AD-code},
this code cannot correct AD errors of weight three.

 \begin{align*}
A_0=B_0 &=
1 - 9\gamma/2 + 153\gamma^2/16 \\
&\quad - 399\gamma^3/32 + 351\gamma^4/32 +  O(\gamma^5); \\
A_i = B_i &=0, \quad i=1,2,4,5,7,8;\\
A_3 = A_9 &=0;\\
B_3 & = 3\gamma^3/4 + O(\gamma^4);\\
B_9 & =  \gamma^9/32;\\
A_6 & =  3\gamma^6/16 + - 9\gamma^7/32 + 45\gamma^8/256 + O(\gamma^9);\\
B_6 & =  3\gamma^6/16  - 9\gamma^7/32  + 9\gamma^8/32   + O(\gamma^9).
\end{align*}
Since the  leading order of $B_3-A_3$  in $\gamma$ is cubic, the Shor code cannot be a $(3,c)$-AD code for any $c>0$.
Hence, this is consistent with the fact that the Shor code corrects two AD errors~\cite{gottesman-thesis}.
\ee

\section{Auxiliary weight enumerators} \label{sec:aux}
Without the existence of a MacWilliams identity, we can nonetheless establish a linear relationship between $|A \>$ and $|B \>$ by introducing additional vectors that reside on an auxiliary space.
We call these vectors \textit{auxiliary exact weight enumerators} or \textit{ auxiliary weight enumerators} for short.
Now the projector $P$ of a quantum code, when decomposed in the Pauli basis, can be written as
\begin{align}
P=&  \sum_{\sigma\in {\cG}_n} 
\frac{\tr(\sigma P)}{2^{n}} \sigma. \label{eq:P}
\end{align}
Our auxiliary weight enumerator depends on the Pauli decomposition~\eqref{eq:P}, and is given by 
\begin{align}
    \psiew
    &= |\phi\> \otimes |\phi\>  ,
\end{align}
where
\begin{align}
|\phi\> = \sum_{\sigma  \in  {\cG}_n} \tr(\sigma P)  |\sigma\> .
\end{align}
The auxiliary weight enumerator is exact in the sense that it encompasses complete information about the quantum code's projector. 
We emphasize that the state $|\phi\>$ depends only  on the code's projector $P$.
Hence $|\phi\>$ is independent of the parameters of the AD channel. 
Since $\tr(\sigma P)\tr(\tau P)$ is invariant under the swap of $\sigma$ and $\tau$, it follows that \begin{align}
\Pi_{} \psiew &= \psiew   ,
\end{align} 
where
\begin{align}
    \Pi_{} 
    &= \sum_{\sigma, \tau \in {\cG}_n}
    |\sigma\>\<\tau| \otimes |\tau\>\<\sigma|.
\end{align}
We later exploit this permutation symmetry to introduce additional constraints in our linear programming bound for amplitude damping channels.

\section{Connection Matrices}\label{sec:connection-matrices}

To establish the connection between our auxiliary weight enumerator with the two generalized weight enumerators, we define matrices that relate the $A$- and $B$-type generalized weight enumerators with the auxiliary weight enumerator as follows:
\begin{align}
M_A\supewe &= 
    \sum_{i=0}^n \sum_{E \in \cK_i}
    \sum_{\sigma , \tau \in {\cG}_n}
    2^{-2n}
    \tr(E \sigma )
    \tr(E^\dagger \tau)
    |i\> \<\sigma|\<\tau|,\label{eq:MA-defi}\\
M_B \supewe 
    &= 
    \sum_{i=0}^n   \sum_{E \in \cK_i}
    \sum_{\sigma , \tau \in {\cG}_n}
    2^{-2n}
    \tr(E \sigma E^\dagger \tau)
    |i\> \<\sigma|\<\tau|
    \label{eq:MB-defi}.
\end{align}
The matrices $M_A$ and $M_B$  establish an indirect linear relationship between the generalized enumerators $|A\>$ and $|B\>$  via an additional linear relationship with the auxiliary weight enumerator. Namely, we have the following linear relationships.
\begin{lemma} 
\label{lem:connection-matrices}
The following matrix identities hold.
\begin{align}
M_A \supewe  \psiew &=(\tr P )^2 |A \>, 
\label{eq:MAewe}\\
M_B \supewe \psiew &= \tr P    |B \>. \label{eq:MBewe}  
\end{align}
\end{lemma}
\begin{proof} 
By expanding the code projector $P$ in the Pauli basis,
we get
\begin{align}
    |A  \>   =&
    \frac{1}{(\tr P )^2}
    \sum_{i=0}^n
    \sum_{E \in \cK_i} 2^{-2n}
    \sum_{\sigma, \tau \in {\cG}_n} 
    \tr(E\sigma)\tr(E^\dagger \tau)
    \notag\\
    &\quad \times
    \tr(\sigma P) \tr( \tau P) |i\>.
\end{align}
Also we can see that 
\begin{align}
    M_A \supewe \psiew
    =
    \sum_{i=0}^n
    \sum_{\substack{
    E \in \cK_i\\
    \sigma , \tau \in {\cG}_n
    }}
  \frac{  \tr(E \sigma )
    \tr(E^\dagger \tau) }{2^{2n} }
    |i\>  
    \tr(\sigma P) 
    \tr(\tau P)  . 
\end{align}
Hence \eqref{eq:MAewe} holds. 

To obtain the second result,
we also expand the code projector $P$ in the Pauli basis to get
\begin{align}
    |B \>   &=
    \frac{1}{\tr(P)}
    \sum_{i=0}^n
    \sum_{E \in \cK_i} 2^{-2n}
    \sum_{\sigma, \tau \in {\cG}_n} 
    \tr(E\sigma E^\dagger \tau)
    \tr(\sigma P) \tr( \tau P) |i\>.
\end{align}
Next, note that 
\begin{align}
    M_B \psiew
    =
    \sum_{i=1}^n     \sum_{\substack{
    E \in \cK_i\\
    \sigma , \tau \in {\cG}_n
    }}
    2^{-2n}
    \tr(E \sigma E^\dagger \tau)
    |i\>  
    \tr(\sigma P) 
    \tr(\tau P)  .
\end{align}
The result 
$M_B \supewe  \psiew = \tr(P)    |B \>$ then follows.
\end{proof}
While we do not have a direct linear relationship between the generalized quantum weight enumerators $|A\>$ and $|B\>$, Lemma~\ref{lem:connection-matrices} establishes a linear relationship between each generalized quantum weight enumerator and the auxiliary weight enumerator. This thereby establishes an indirect linear relationship between $|A\>$ and $|B\>$, which later allows us to obtain linear programming bounds for amplitude damping codes. 

It is also important to note the following properties of connection matrices.
\begin{enumerate}
    \item The connection matrices $M_A\supewe$ and $M_B\supewe$ are devoid of information content about the code, because they are both independent of the code projector $P$.
\item The connection matrices $M_A\supewe$ and $M_B\supewe$ depend on the damping parameter $\gamma$ of AD  channel.
\end{enumerate}

\section{Linear Programming Bounds} \label{sec:linear program bounds}

From the above discussion, the weight enumerators $\ket{A}$ and $\ket{B}$ of a $(t,c)$-AD code must satisfy (\ref{eq:BleqA}), (\ref{eq:MAewe}), and (\ref{eq:MBewe}). We formulate a linear program with a constant objective function, and find non-negative variables $A_0,\dots,A_n$, $B_0,\dots,B_n$ that belong to a particular feasible region. 
The feasibility problem of our linear program is then equivalent to the following. 
\begin{align}
    {\rm Find}\ \ A_0,\dots,&A_n,B_0,\dots, B_n\notag\\
      {\rm subject\ to\ } 
(\tr P)^2|A  \>      &=M_A \supwe \psiew \notag\\
\tr P|B  \> &=M_B \supwe \psiew\notag\\
(B_i    - A_i)/\gamma^{t+1}  &\le c  , \quad 0 \le i \le d-1\notag\\ 
B_i    /\gamma^i  &\le \binom n i  , \quad 0 \le i \le n\notag\\ 
B_0+\dots +B_n      &\le 1 \notag\\ 
A_i      &\ge 0, \quad 0 \le i \le n  \notag\\
\Pi_{} \psiew &= \psiew. \label{eq:linear-constraints}
\end{align}
Since integer programs are hard to solve in general, 
our feasibility conditions are attractive because they have no integer constraints, in contrast to many other linear programming bounds for stabilizer codes \cite{CRSS98,LBW13,LA18,ALB20}. 
Hence, we have a linear program as opposed to an integer program.
However, one may wonder whether such a linear program is sufficiently constrained to be potentially infeasible.
We demonstrate numerically that our linear program can be infeasible, by analyzing the potential of using three qubits to correct a single AD error.
To do this, we have an additional observation that our linear program is parameterized by $\gamma$.
Since a $(t,c)$-AD code is defined for an arbitrary $\gamma$, we can concatenate the linear constraints for various values of $\gamma$. 
Crucially, constraints for different values of $\gamma$ are related because $\psiew$ is independent of $\gamma$. We illustrate the linear dependence of all of our linear constraints in Fig.~\ref{fig:star}.

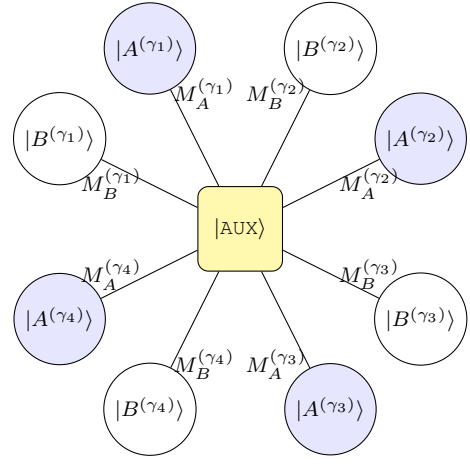
\begin{figure}
\begin{center}
		\begin{tikzpicture}[scale=1.2][thick]
		\fontsize{8pt}{1} 
		\tikzstyle{checknode} = [draw,fill=blue!10,shape= circle,minimum size=0.4em]
		\tikzstyle{variablenode} = [draw,fill=white, shape=circle,minimum size=0.4em]
		\tikzstyle{measurenode} = [draw,fill=yellow!40,shape= rectangle,rounded corners, minimum size=4.0em]
		\node[measurenode] (mn1) at (0,0) {$\ket{{\texttt{AUX}}}$} ;
		\node[checknode] (cn1) at (-1,2) { $\ket{A^{(\gamma_1)}}$} ; ;
		\node[variablenode] (vn4) at (-1,-2) { $\ket{B^{(\gamma_4)}}$} ;
		\node[variablenode] (vn2) at (1,2) {$ \ket{B^{(\gamma_2)}}$} ;
		\node[variablenode] (vn3) at (2,-1) {$\ket{B^{(\gamma_3)}}$} ;
			\node[checknode] (cn3) at (1,-2) { $\ket{A^{(\gamma_3)}}$} ; ;
		\node[checknode] (cn4) at (-2,-1) { $\ket{A^{(\gamma_4)}}$} ;
		\node[checknode] (cn2) at (2,1) {$ \ket{A^{(\gamma_2)}}$} ;
		\node[variablenode] (vn1) at (-2,1) {$\ket{B^{(\gamma_1)}}$} ;
		\node (label1) at (-0.4,1.5) {$M_A^{(\gamma_1)}$} ;
	    \node (label2) at (0.4,1.5)  {$M_B^{(\gamma_2)}$}; 
	    \node (label1) at (-0.4,-1.5) {$M_B^{(\gamma_4)}$} ;
	    \node (label2) at (0.4,-1.5)  {$M_A^{(\gamma_3)}$}; 
		%
		\draw (cn1) --  (mn1) (vn1) --node[midway,left]{$M_B^{(\gamma_1)}$} (mn1)  (vn2) --(mn1) (cn2) --node[midway,right]{$M_A^{(\gamma_2)}$} (mn1);
		\draw (cn3) --  (mn1) (vn3) --node[right]{$M_B^{(\gamma_3)}$} (mn1)  (vn4) -- (mn1) (cn4) --node[midway,left]{$M_A^{(\gamma_4)}$} (mn1);
		%
		\end{tikzpicture}
\end{center}
	\caption{
	   The relationship between various enumerators is depicted here. Every A and B-type enumerator for differing values of AD parameter $\gamma_i$ relates linearly to the same auxiliary enumerator.
	  }\label{fig:star}
\end{figure}

To determine if our concatenated linear program is feasible, we code up the linear constraints in the \texttt{Matlab} solver \texttt{cvx}, and use the algorithm SDPT3.  
In the linear constraints of \eqref{eq:linear-constraints}, we write the monomials of $\gamma$ as denominators. This normalizes our constraints so that a numerical solver can be numerically stable even for small values of $\gamma$.
Also, when coding up the linear constraints of \eqref{eq:linear-constraints} in a solver, 
we do not explicitly construct the permutation matrix $\Pi$ because it is much too big.
Rather we specify its implied linear constraints directly into the optimizer environment for our linear program.

In our numerical study, we analyze the possibility of correcting a single AD error using three qubits.
We obtain mainly no-go results on the existence of a three-qubits code that corrects a single AD error.
For this, we consider four different values of $\gamma$ in the construction of our linear program.
More precisely, we numerically find the maximum $c$ for which the convex solver returns a result that says that the linear program is infeasible. 

\be
For $n=3$, $M=2$, $t=1$, we rule out $c=9.8 \times 10^4$ using
$\gamma =  0.1 ,   0.05,    0.01,    0.0001$.
\ee

\section{Discussions} \label{sec:discussions}
In this paper, we showed that quantum weight enumerators can be generalized to the setting of AQEC AD codes. Key to our analysis is our introduction of auxiliary weight enumerators, which allows us to establish an indirect linear relationship between the generalized quantum weight enumerators.

As it stands, the auxiliary weight enumerator is a vector of size $4^{2n}$ in the number of qubits $n$.
If we restrict our attention to stabilizer codes, this size potentially can be greatly reduced. This is because the representation of the code projectors of stabilizer codes in the Pauli basis can be written entirely in terms of the code's stabilizers. Since the number of stabilizers for an $[[n,k]]$ code is $2^{n-k}$, this number is far fewer than $4^{2n}$.
In view of this, we will discuss the extent in which the connection matrices and auxiliary weight numerators can be compressed in a subsequent work. This will make tractable exploration of the performance of larger sized codes. 

\section{Acknowledgements}\label{eq:acknow}
YO acknowledges support from the EPSRC (Grant No. EP/M024261/1) and the QCDA project (Grant No. EP/R043825/1)) which has received funding from the QuantERA ERANET Cofund in Quantum Technologies implemented within the European Union’s Horizon 2020 Programme.
CYL was financially supported from the Young Scholar Fellowship Program by Ministry of Science and Technology (MOST) in Taiwan, under Grant No. MOST108-2636-E-009-004.


\end{document}